\documentclass{article}
\usepackage{float}
\usepackage{amsthm}
\usepackage{amsmath}
\usepackage{amssymb}
\usepackage{graphicx}
\usepackage[authoryear]{natbib}
\usepackage[unicode=true,pdfusetitle,
 bookmarks=true,bookmarksnumbered=false,bookmarksopen=false,
 breaklinks=false,pdfborder={0 0 1},backref=section,colorlinks=false]
 {hyperref}

\makeatletter

\floatstyle{ruled}
\newfloat{algorithm}{tbp}{loa}
\providecommand{\algorithmname}{Algorithm}
\floatname{algorithm}{\protect\algorithmname}

\numberwithin{equation}{section}
\numberwithin{figure}{section}
\theoremstyle{plain}
\newtheorem{thm}{\protect\theoremname}
\theoremstyle{definition}
\newtheorem{defn}[thm]{\protect\definitionname}
\theoremstyle{plain}
\newtheorem{cor}[thm]{\protect\corollaryname}
\ifx\proof\undefined
\newenvironment{proof}[1][\protect\proofname]{\par
\normalfont\topsep6\p@\@plus6\p@\relax
\trivlist
\itemindent\parindent
\item[\hskip\labelsep
\scshape
#1]\ignorespaces
}{%
\endtrivlist\@endpefalse
}
\providecommand{\proofname}{Proof}
\fi
\theoremstyle{plain}
\newtheorem{lem}[thm]{\protect\lemmaname}

\usepackage{authblk}

\usepackage{times}% For figures
\usepackage{subfigure}

\usepackage{algorithm}\usepackage{algorithmic}

  \providecommand{\corollaryname}{Corollary}
  \providecommand{\definitionname}{Definition}
  \providecommand{\lemmaname}{Lemma}
\providecommand{\theoremname}{Theorem}

\makeatother

\providecommand{\corollaryname}{Corollary}
\providecommand{\definitionname}{Definition}
\providecommand{\lemmaname}{Lemma}
\providecommand{\theoremname}{Theorem}

\title{Localized epidemic detection in networks with overwhelming noise}
\author[1]{Eli A. Meirom} 
\author[2]{Chris Milling} 
\author[2]{Constantine Caramanis} 
\author[1]{Shie Mannor} 
\author[1]{Ariel Orda} 
\author[2]{Sanjay Shakkottai}
\affil[1]{Department of Electrical Engineering, Technion-  Israel Institute of Technology,  Israel } 
\affil[2]{Dept of Electrical and Computer Engineering, The University of Texas at Austin, USA}

\date{}

\begin{document}

\maketitle

\begin{abstract}
We consider the problem of detecting an epidemic in a population where
individual diagnoses are extremely noisy. The motivation for this
problem is the plethora of examples (influenza strains in humans,
or computer viruses in smartphones, etc.) where reliable diagnoses
are scarce, but noisy data plentiful. In flu/phone-viruses, exceedingly
few infected people/phones are professionally diagnosed (only a small
fraction go to a doctor) but less reliable \emph{secondary signatures}
(e.g., people staying home, or greater-than-typical upload activity)
are more readily available.

These secondary data are often plagued by unreliability: many people
with the flu do not stay home, and many people that stay home do not
have the flu. This paper identifies the precise regime where knowledge
of the\emph{ contact network} enables finding the needle in the haystack:
we provide a distributed, efficient and robust algorithm that can
correctly identify the existence of a spreading epidemic from highly
unreliable local data. Our algorithm requires only local-neighbor
knowledge of this graph, and in a broad array of settings that we
describe, succeeds even when false negatives and false positives make
up an overwhelming fraction of the data available. Our results show
it succeeds in the presence of partial information about the contact
network, and also when there is not a single ``patient zero,'' but
rather many (hundreds, in our examples) of initial patient-zeroes,
spread across the graph. 
\end{abstract}
\label{submission}

\section{Introduction}

Any study, analysis or curbing of an epidemic begins with the fundamental
question: is the phenomenon that we experience indeed an epidemic?
The identification of such processes, be they malicious malware spreading
on computer networks, memes on social networks, or trends in public
opinion, is of major interest across several disciplines. 

A common characteristic is a stark absence of reliable information:
patients with flu-like symptoms rarely visit a medical professional.
More typical, however, is the availability of significant amounts
of information indicating a secondary signature of the outbreak: flu
patients may stay home from work, infected computers may exhibit somewhat
encumbered performance, adopters of a new technological trend (e.g.,
smart-watch) may show new usage patterns, or may simply ``self-report.''
While more plentiful, such data may be riddled with false positives
and negatives.

Further complicating the problem is the fact that the underlying network
is rarely fully known and in practice it may be possible to recover
only a very limited, local, subgraph near any infected node. In the
setting of an epidemic, there may be multiple epidemic sources rather
than a single ``patient zero.'' In the early phase of the HIV epidemic,
this was precisely the nature of the data, since ``patient zero''
was unknown, and a ``hidden backbone'' connected the initially identified
patients (\citet{Auerbach1984}).

This paper addresses the algorithmic and statistical challenges that
this problem poses on large scale networks. We consider the basic
robust graph-learning problem in the most dire information-restricted
setting: at some instant in time we are informed that a given subset
of nodes exhibiting unusual behaviour (in the computer setting) or
are sick. Exact reporting times are unaccessible, and moreover we
assume we are unable to observe the time evolution of the sickness
reporting process. Given this single snapshot in time, the statistical
inference problem is to determine if there is an epidemic, spreading
in a diffusive process, which propagates through the network from
one or possibly many initial nodes, or rather nodes have become infected
via an independent, external mechanism.

Algorithmically, we seek an efficient, scalable algorithm with minimal
computational and information requirements. In particular, we assume
each node has some knowledge of its neighbourhood (our simulations
show that this too need only be approximate).

Often, the two processes of random infection and epidemic spread,
exist in parallel. For example, a video may be posted by multiple
members in a social network, and spread across the network by sharing;
technology trends may be spread by mass-media advertisement, as well
as word-of-mouth effects. We discuss this possibility explicitly,
and we ask which proliferation mechanism dominates: are more people
exposed to this media by their friends, or is it more common to watch
it first through an external website.

We describe a class of algorithms for this decision making, or statistical
hypothesis testing problem. Our algorithm requires only local information,
and it is computationally efficient. Furthermore, the algorithm can
be applied in a distributed fashion. Our theoretical results show
that it works even when the fraction of false negatives and positives
goes to 1 -- i.e., even as an overwhelming majority of information
we see is in fact false. Our simulation results corroborate this finding.

\section{\label{sec:Problem-definition,-related}Problem definition, related
work and our contribution}

In this section we describe the basic model we consider, review related
work, and then explain our contributions.

\subsection{The Basic Model}

The problem we consider is an extension of the scenario described
in \citet{Milling2012,Milling2013}. Consider a graph $G=(V,E)$,
where the number of nodes is $N=|V|.$ The decision making task at
hand is distinguishing between the two following alternative scenarios. 
\begin{itemize}
\item \emph{An epidemic}: At time $t=0$, an arbitrary subset of nodes is
infected. The infection then spreads according to a standard \emph{suceptible-infectedmodel}
(\citet{Ganesh2005}). That is, infected nodes infect their neighbors
according to an exponential clock set on each edge. We denote the
set of infected nodes at time $T$ as $S=S(T)$, and the infection
size as $|S|=\alpha(N)N$. At this time, each infected node \emph{reports
it is infected} with probability $q(N)$. We denote the set of truly
reporting nodes by $S_{r}\subseteq S$, and thus call $S\setminus S_{r}$
the \emph{false negatives}. In addition, there are $f|S_{r}|$ nodes,
picked uniformly over the network, that also report infected. The
intersection of these nodes with $S^{c}$, represent false positives.
A reporting process illustration is presented in Fig \ref{fig:reporting illustration}. 
\item \emph{Uniform reporting}: Each node, independently of all others,
reports an infection with some probability. The problem is most interesting
when this probability is chosen so that the expected numbers of reporting
nodes in each scenario match. 
\end{itemize}
Thus in summary, the key parameters that define our setting are as
follows: $N$ is the total number of nodes; $\alpha(N)$ is the fraction
of nodes ultimately infected, denoted by $S$; $q(N)$ is the probability
that a node in $S$ will report, and hence $(1-q)$ is the false negative
rate; $f$ controls the fraction of false positives, which scale as
$f/(1+f)$, and hence goes to $100\%$ as $f\rightarrow\infty$.

Many settings exhibit the above structure. Influenza and allergies
are known to produce similar symptoms -- without a professional diagnosis,
these may often be confused, even more so when the only indication
of either is absence from work or school. Influenza of course is the
epidemic, whereas allergies affect near and distant neighbors independently.
Technology adoption shares similar traits. Word of mouth advertising
spreads like an epidemic, where as mass media advertisement affects
customers independently.

\begin{figure}
\centering{}\includegraphics[width=0.7\columnwidth]{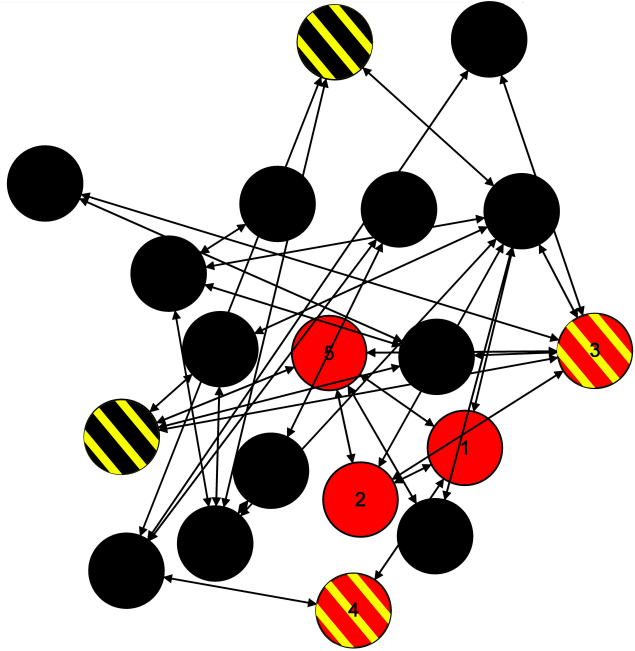}\caption{\label{fig:reporting illustration}An illustration of the epidemic
reporting process. The contagion spans the red nodes, numbered by
the infection order, starting from ``patient zero'', numbered by
one. The infected network fraction is $\alpha=0.25$ . There are two
truly reporting nodes, $\left(q=0.4\right)$ denoted by red and yellow
stripes. In addition, there are two more false positives $\left(f=1\right)$,
in black and yellow stripes.}
\end{figure}

A sample reporting map of the two processes is displayed in Fig.\,\ref{fig:Reporting-map-samples}.
When the reporting probability increases, $q(N)\rightarrow1$, there
exists a large connected component with a ball-like shape about the
set of ``patient-zero''s and the problem is easier. Likewise, the
problem is more difficult when $f(N)\rightarrow\infty$, as the truly
reporting nodes are washed out by the sea of falsely reporting nodes.
We describe a class of algorithms that is shown to converge correctly
even in settings where $q(N)\rightarrow0$ and $f(N)\rightarrow\infty$.

\begin{figure}
\centering{}\includegraphics[width=1\columnwidth]{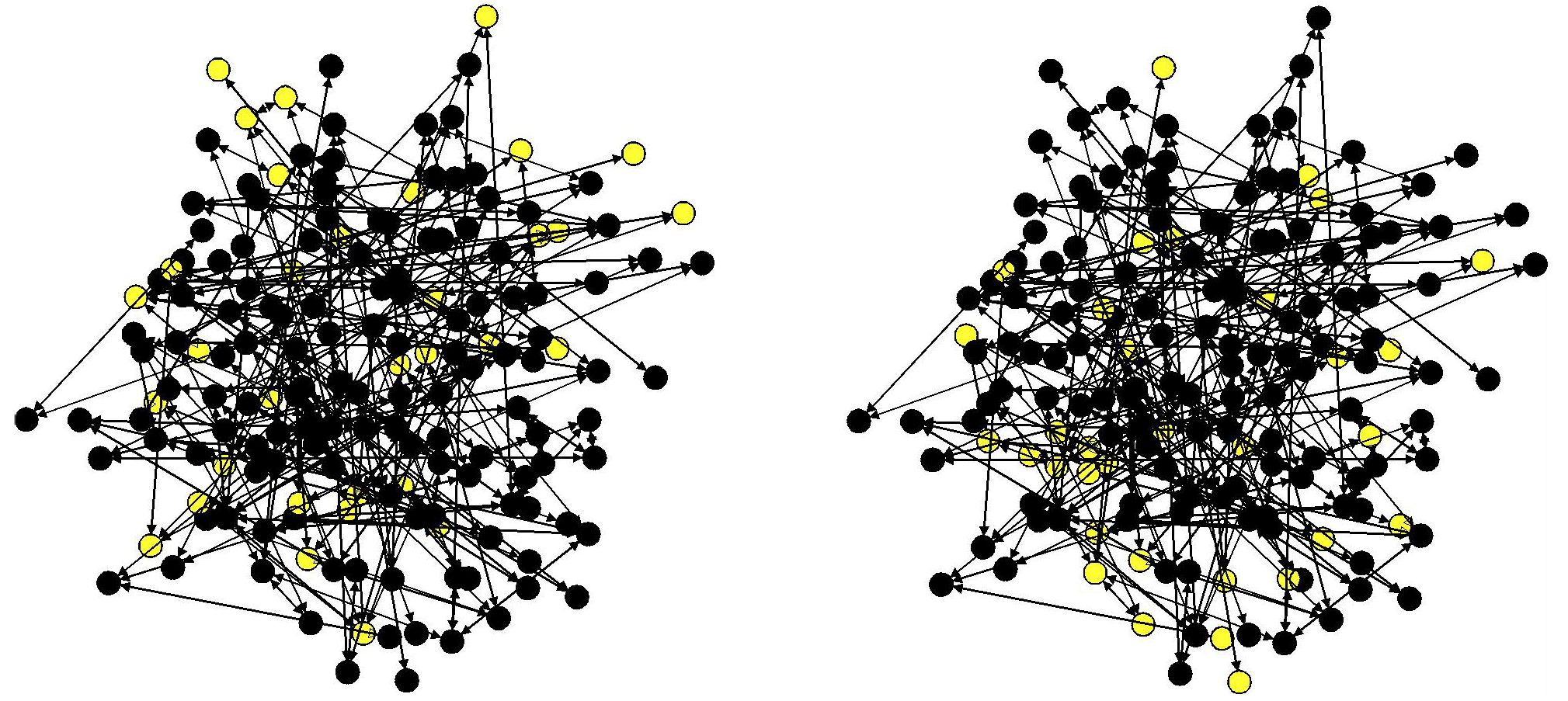}\caption{\label{fig:Reporting-map-samples}Reporting map samples. The reporting
nodes are in yellow. (left) A reporting sample in a uniform reporting
setting. (right) An epidemic reporting scenario map. It is difficult
to distinguish between the two hypotheses visually. }
\end{figure}

\subsection{Prior Work}

The predictive (forward) analysis of our SI model, as in (\citet{Ganesh2005}),
is tightly related to results in First Passage Percolation theory
(\citet{Blair-Stahn2010}). These and related works focus on modeling
the epidemic spread characteristics, for example, estimating the infection
rate for different topologies (e.g., \citet{Gopalan2011}), predicting
the first time the infection exceeds a given boundary, and so on.
These are, therefore, prediction problems, or \emph{forward problems}:
given the initial conditions, predict what happens in the future.

Our work focuses on the inverse question: given (a noisy version of)
what happened, solve the inference problem to decide if the process
is an epidemic or a random illness. Recently, related inference questions
have gained considerable attention. Of particular interest are estimating
the model parameters, such as transmission rate, either by MCMC methods
(\citet{Streftaris}), or a Bayesian approach (\citet{Demiris}).
Another frequently discussed question is the identification of the
epidemic source (e.g., \citet{Shah,Karamchandani2013}).

The work that most closely relates to ours paper is \citet{Milling2012,Milling2013}.
They consider a similar model, and seek to solve the same problem.
The key differences here are that the algorithm and analysis are completely
different, with some important consequences. In those works, the authors
rely on tools from first passage percolation, which in some settings
(for some graph topologies, for initial infection conditions) may
be fragile. For example, it is not clear if the algorithms given there
are able to handle large numbers of initial seeds. Our results in
Section \ref{sec:Experiments} indicate that our algorithm is largely
insensitive to the number of initial infections, and works well for
even hundreds of spread out ``patients zero.'' A second important
difference comes on the \emph{algorithmic front}. Our algorithm is
inherently local. As we describe below, it essentially counts nodes
with many infected neighbors, at a high level bearing similarity to
triangle-counting as a proxy for graph clustering (\citet{Watts1998}).
This allows our algorithm to be run in parallel, and most significantly,
requires only local knowledge of the graph -- something that may be
critical in applications. In Section \ref{sec:Experiments} we compare
the performance of our new algorithm, and find that both on synthetic
(e.g., Erdos-Renyi) graphs as well as real-world networks (e.g., part
of the Facebook graph) our algorithm produces statistically much better
results, and significantly less computational cost. As we explain
below, our algorithm's running time scales with the number of reporting
infected nodes, which may be far smaller than the number of total
nodes.

\subsection{Our Contributions}

There are two factors which control the ``difficulty'' of the hypothesis
testing problem we have posed. First, the more the false negatives
(i.e., the smaller the set $S_{r}$), the more the epidemic statistically
resembles the random illness. Second, in the limit when all infected
nodes report, the presence (or absence) of a large connected component
is easy to test, and sufficient to solve the decision problem. Finally,
the topology of the graph itself plays an important role. The more
connected the graph is, the more the two processes become statistically
identical. In the limiting case where the contact network is the complete
graph, the two processes are indistinguishable.

From the algorithmic standpoint, there is a further regime of interest:
the very small infection setting. It is interesting to understand
in this setting, how well an algorithm using only minimal graph information,
can perform.

Our contributions are statistical and algorithmic, and address both
regimes discussed above. Specifically, the key contributions of this
work are as follows. 
\begin{itemize}
\item Algorithm: We provide a simple distributed algorithm that runs in
time linear in the number of reporting infected nodes, can be easily
parallelized, and requires minimal knowledge of the graph: each infected
node is required to know only information about its local neighborhood. 
\item Statistical Performance: We give sufficient conditions for our algorithm
to correctly identify the infection cause (epidemic or random). In
particular, we give conditions that guarantee that our algorithm succeeds
even when the number of infected nodes is a $O(N)$, i.e., a constant
fraction of all the nodes. We also consider the small-infection regime.
We show here that our algorithm succeeds \emph{even if each node only
knows the graph up to its 1-hop neighbors}. 
\item Experiments: we provide experiments on synthetic (random) graphs,
as well as real-world graphs. We show that our algorithm's performance
is at least as good as the theory predicts. Moreover, its efficiency
allows it to scale easily to very large networks. We compare to state-of-the-art
and demonstrate that our algorithm is both faster and more accurate. 
\end{itemize}

\section{The Algorithm: Hotspot Aggregator}

The intuition behind our algorithm is simple: in an epidemic, having
sick neighbors is more common than in the case of a random illness. 
\begin{defn}
For $B_{i}(l)$ denoting the radius-$l$ ball about node $i$, and
$NN_{i}(K)$ denoting the set of the $K$ nearest neighbors of node
$i$, we define the following indicator random variables:

\[
H_{i}^{NN}(K,s)=\begin{cases}
1 & \left|\{v\in NN_{i}(K)\cap S\}\right|\geq s\\
0 & \text{otherwise}
\end{cases}
\]

\end{defn}
and a similar indicator, $H_{i}^{B}(l,s),$ for the radius-$l$ ball$.$
These indicators are true only if there are enough (more than $s$)
reporting nodes in the immediate neighborhood of $i$. These are equivalent
with appropriate correspondence of $l$ and $K$, and so we use them
interchangeably. If the indicator evaluates to 1, we call node $i$
a \emph{hotspot}. The above intuition suggests that there are more
hotspots in an epidemic than in a random infection. Thus the number
of hotspots, i.e., the sum of these indicators over reporting infected
nodes, suggests a threshold test, that can correctly classify the
infection cause. This is precisely the crux of Hotspot Aggregator
Algorithm, Algorithm \ref{alg:The-hotspot-aggregator}.

\begin{algorithm}
\textbf{Input}: Threshold values $K,\: T.$

\textbf{Output}: An epidemic or a random reporting scenario

\begin{algorithmic} 
\STATE ${\rm counter} \leftarrow 0$  
\FOR {$i=1$ to $N_{\mbox{{\tiny{reporting}}}}$}   
\STATE ${\rm counter} \leftarrow {\rm counter}+H^{NN}_{i}(K,K)$  
\ENDFOR  
\IF{${\rm counter}>T$}  
\STATE \textbf{return} epidemic  
\ELSE \STATE \textbf{return} uniform reporting  
\ENDIF  
\end{algorithmic}

\caption{\label{alg:The-hotspot-aggregator}The hotspot aggregator}
\end{algorithm}

\textbf{Theoretical/Practical Implementation}. The algorithm depends
on the nearest-neighbor threshold value $K$, and the parameter $T$.
The theorems that make up our main results in the next section, demand
a careful choice for $K$ and $T$ that requires approximate knowledge
of the false positive rate and the graph topology. As we detail in
the appendix, for some topologies the values for $K$ and $T$ can
be computed analytically. More generally, the optimal values can be
computed through simulation. On the other hand, our experimental results
in Section \ref{sec:Experiments} suggest that the algorithm's performance
is quite stable. For instance, choosing $K$ corresponding to the
$r$-hop neighborhood in a graph, for $r=2,3$ yields uniformly good
computational results (see Section \ref{sec:Experiments} for more).

\textbf{Algorithm Information and Complexity}. The complete graph
information (its global structure) requires $O\left(N^{2}\right)$
to encode. In contrast, our algorithm requires only the vector $C$,
which contains the number of reporting nodes within each local environment.
The length of $C$ is the number of reporting nodes, denoted by $N_{\mbox{{\tiny{reporting}}}}$,
where for many cases of interest, $N_{\mbox{{\tiny{reporting}}}}$
can be as small as $O(\log(N))$ or even lower. In addition to being
a parsimonious representation of the information required by our algorithm,
the vector $C$ may be reconstructed by knowledge of only the local
neighborhood of the reporting nodes, rather than the whole graph.
As discussed above, such local information is often the only reasonably
accessible/reliable information.

The algorithm's complexity is $o(N_{\mbox{{\tiny{reporting}}}})$
and it scales with the number of reporting nodes (which, again, is
typically vanishingly small compared to the to the number of nodes
in the network). Similarly, the memory requirement is also very small
and therefore this algorithm is easily scalable.

\textbf{Algorithm Performance}. As the noise increases, so does the
difficulty of the decision problem . In the next section we discuss
the convergence properties of the algorithm under extremely noisy
conditions and present specific choices for the parameters $T$ and
$K$. We show that under relaxed topological constraints, the algorithm
converges correctly even when the number of falsely reporting nodes
is $\Theta(N)$, irrespective of the number of truly reporting nodes,
which may be even $\omega(1)$. In an extreme scenario, where the
number of false positives is $\Theta(N)$ greater than the truly reporting
nodes number, Corollary \ref{cor:large number of reporting} shows
the local environment about every reporting nodes should contain $\Theta(\log N)$
nodes. However, if the number of falsely reporting nodes is $o(N)$,
then it possible to apply the algorithm on a local environment which
contains only a finite number of nodes, as shown in Theorem \ref{prop: small epidemics converges}.

\section{\label{sec:Correctness-and-convergence} Main Results: Correctness
\& Convergence}

Before we proceed, we require some preliminary definitions.The infection
boundary is the set of infected nodes that have a non-infected node
in their local neighborhood. An alternative definition for a boundary
is a set of infected nodes for which at least one of the $k$ nearest
neighbors is not infected.
\begin{defn}
The $l$th order boundary $\partial S(l)$ is a subset of infected
nodes, $\partial S(l)\subseteq S$, such that for every $i\in\partial S(l)$
we have $B_{i}(l)\cap S^{c}\neq\emptyset$, i.e., there exists $j\notin S$
in the ball $B_{i}(l)$. Alternatively, the $k$th order border set
is a subset of infected nodes, $DS(k)\subseteq S$, such that for
every $i\in DS(k)$ we have $NN_{i}(k)\cap S^{c}\neq\emptyset$. We
denote by $\gamma(k)=1-\left|DS(k)\right|/|S|$ the fraction of interior
nodes, i.e., nodes for which their local neighborhood is contained
in the infected region.
\end{defn}
As discussed in the introduction, this decision problem's difficulty
increases with the number of false positives and false negatives.
The most challenging setting is when the number of false positives
is $\Theta(N)$, regardless of the number of truly reporting nodes,
which may be $\omega(1).$ We call this the \emph{dense regime}. As
shown in our first theorem, in this regime, if $\gamma(K)$ is non
zero for $K=O(\log N)$, the algorithm converges correctly. Alternatively,
assume the number of truly reporting node is also $\Theta(N)$. If
$\gamma\left(K\in O(1)\right)$ is non zero then the algorithm converges
as well.

The algorithm succeeds when the number of hotspots in the epidemic
and random scenarios differ significantly. The hotspot indicator is
a Bernoulli random variable, hence computing the expectation of their
sum across reporting nodes is straightforward. However, these Bernoulli
random variables are correlated, in proportion to their graph distance
to each other. Thus the key technical portion of the proof involves
controlling the variance of this sum, by appropriately harnessing
large deviation results for sums of correlated random variables.
\begin{thm}
\label{prop:non-vanishing boundary}Assume the number of reporting
nodes, whether truly reporting or falsely reporting, is $\Theta(N)$.
Assume there exist values $K,\gamma$ such that

a) $\Pr\left(\gamma(K)\neq0\right)\rightarrow1$ for $N\rightarrow\infty$

b) \textup{$K=\left\lceil \log(\gamma^{-1}\left(f+1\right))\right\rceil .$}

Then, the hotspot aggregator algorithm with parameters $K$ and 
\begin{eqnarray*}
T & = & \frac{N_{\mbox{{\tiny{reporting}}}}}{2}\left(\frac{\gamma p_{in}^{K}}{\left(f+1\right)}+p^{K}\right)
\end{eqnarray*}
classifies correctly with high probability. The type I error and type
II error decay rates are $o\left(\exp(-cN_{\mbox{{\tiny{reporting}}}})\right)$,
where $c$ is a constant that depends on the problem parameters. 
\end{thm}
As an immediate corollary of this result, we see that in the large-infection-regime,
our algorithm succeeds even as the numbers of false negatives and
positives are overwhelmingly greater than the number of truly infected
reporting nodes. The proof follows immediately from the theorem, and
the details are deferred to the appendix.
\begin{cor}
\label{cor:large number of reporting} For a large infection with
$\Theta(N)$ infected nodes (i.e., $\alpha=\mbox{\ensuremath{\Theta}(1) }),$
the algorithm converges correctly if conditions (a) and (b) above
are satisfied, even when the reporting probability goes to zero and
noise ratio goes to infinity, i.e., $q\in\omega\left(N^{-1}\right)$
and $f\in\omega\left(q^{-1}\right)$. 
\end{cor}
In particular, if the number of truly reporting nodes is $\Theta(N)$,
then the algorithm converges correctly if $\gamma(K=const)>0$. Alternatively,
the algorithm converges correctly for every network such that $\gamma(K=2\log N)>0$.
This holds, for example, for grids and tree like networks, even if
the noise is $f\in\Theta(N)$.
\begin{proof}
Denote the epidemic (uniform reporting) indicator as $\mathcal{I}$
(respectively, $\tilde{\mathcal{I}}$ ), the reporting probability
of an \emph{infected} node in the epidemic scenario by $p_{in}$.and
the node reporting probability in the uniform reporting scenario as
$p$. In the uniform reporting scenario, the probability that a local
hotspot indicator is true is 
\begin{eqnarray*}
\Pr\left(|\{v\in NN_{i}(k)\cap S_{\mbox{{\tiny{reporting}}}}\}|\geq k\right) & = & p^{k}.
\end{eqnarray*}

Therefore, the expected number of hotspots in the uniform reporting
case $\sum_{i\in V}H_{i}|\tilde{\mathcal{I}}$ is 
\[
\mathbb{E}\left(\sum_{i\in V}H_{i}|\tilde{\mathcal{I}}\right)=\sum_{i\in V}\mathbb{E}\left(H_{i}|\mathcal{\tilde{\mathcal{I}}}\right)=N_{\mbox{{\tiny{reporting}}}}p^{k}.
\]

In the epidemic setting, the hotspot number is greater than the interior
(non-boundary) hotspot number. Therefore, 
\[
\mathbb{E}\left(\sum_{i\in V}H_{i}|\mathcal{I}\right)\geq\mathbb{E}\left(\sum_{i\in S,i\notin DS}H_{i}|\mathcal{I}\right)\geq\frac{\gamma N_{\mbox{{\tiny{reporting}}}}}{\left(f+1\right)}p_{in}^{k}
\]

where we used the fact that the expected number of truly reporting
nodes is $N_{\mbox{{\tiny{reporting}}}}/\left(f+1\right)$ and applied
Lemma \ref{lem:binmoial sum} (stated in the appendix).

Set $T=N_{\mbox{{\tiny{reporting}}}}\left(\frac{\gamma p_{in}^{k}}{\left(f+1\right)}+p^{k}\right)/2$.
Applying Lemma \ref{lem: inside outside random probabilities}, choose
\begin{equation}
K\geq\log(\gamma^{-1}\left(f+1\right))\geq\log(\gamma^{-1}\left(f+1\right))p/p_{in}.\label{eq:K}
\end{equation}

Now the key step is to evaluate the error rates, despite the correlation.
To do this, we use a Hoeffding-like large deviation theorem (see Corollary
2.6 in \citet{Janson2004}) for graph-structured correlation decay
patterns. We can form a dependence graph $\Gamma$ among these $N_{\mbox{{\tiny{reporting}}}}$
random variables, drawing an edge between any two non-independent
random variables. Note, then, that for each two nodes, $i$ and $j$,
the hotspot indicators $H_{i}$ and $H_{j}$ are independent iff the
corresponding environments $B_{i}$ and $B_{j}$, or $NN_{i}$ and
$NN_{j}$ are disjoint. The number of nodes in each local environment
of the $k$ nearest neighbors of $i$ is also $k$, so the number
of nodes with disjoint local environments is at least $N-k^{2}$.
In other words, the maximal node degree in $\Gamma$ is $k^{2}$.

Set $P=p^{K},P_{in}=\gamma p_{in}^{k}/(f+1)$. Then, adapting some
concentration inequalities from \citet{Janson2004}, we can show that
the type I error, $E_{I}:=\Pr\left(\sum_{i\in V}H_{i}|\tilde{\mathcal{I}}\geq T\right)$
is bounded by 
\[
E_{I}\leq\exp\left(-N_{\mbox{{\tiny{reporting}}}}\frac{\left(P_{in}-P\right)^{2}}{16\left(K^{2}+1\right)\left(P+\left(P_{in}-P\right)/6\right)}\right)
\]

In the epidemic scenario, the total number of hotspots is at least
as great as the number of interior hotspots. Therefore, the type II
error is bounded by 
\[
E_{II}=\Pr\left(\sum_{i\in V}H_{i}|\mathcal{I}<T\right)\leq\Pr\left(\sum_{i\in S,i\notin DS}H_{i}|\mathcal{I}<T\right)
\]

and similarly, the latter quantity is bounded by 
\[
\exp\left(-N_{\mbox{{\tiny{reporting}}}}\frac{\left(P_{in}-P\right)^{2}}{16\left(K^{2}+1\right)P_{in}}\right)
\]

As shown in the appendix, these errors tend to zero under the conditions
of this theorem.
\end{proof}
\textbf{Applying Theorem \ref{prop:non-vanishing boundary} and Corollary\ref{cor:large number of reporting}}.
As with the algorithm, applying the theorem or \textbf{corollary}
requires calculating values of $\gamma(K)$ and $K$. As we show in
the appendix, $\gamma(K)$ can often be explicitly (analytically)
computed, and if that is not available, it can be easily computed
numerically. In the appendix, we also present explicit instructions
on the application of the hotspot algorithm for general networks,
including specific finite size networks for which the statistical
ensemble is unknown. When $\gamma(K)$ is known, one can simply find
the minimal value of $K$ such 
\[
K\geq\log\left(\gamma(K)\right)+\log\left(f+1\right),
\]
and substitute the corresponding value in the threshold. For example,
for $d$-dimensional grids one can choose $K=\left\lceil \log(f+1)\right\rceil $
and 
\[
T=\frac{N_{\mbox{{\tiny{reporting}}}}}{2}\left(\frac{p_{in}^{K}}{\left(f+1\right)}+p^{K}\right),
\]
while for a tree like network, such as an Erdos-Renyi network, choose
$K$ such that $K\geq\log K+\left\lceil \log(f+1)\right\rceil $ and
\[
T=\frac{N_{\mbox{{\tiny{reporting}}}}}{2}\left(\frac{p_{in}^{K}}{K\left(f+1\right)}+p^{K}\right).
\]

If the function $\gamma(K)$ is not known or if it difficult to describe
analytically and solve for $K,$ one can apply the bound $\gamma(K)\geq1/\alpha N$
in eq. \ref{eq:K} and obtain a corresponding value for the threshold
$T$.

As long as the probability to find a node for which the local environment
is contained in the infected region is non-zero, the algorithm converges
correctly. This is a very lenient requirement, which only fails for
nearly full-mesh graphs, such that with high probability, for \emph{every}
ball, either there are less than $\log N$ nodes, $|B_{i}(l)|\leq\log N$,
or the number of nodes in the ball is infinitely higher than $\log N$,
$\log N/|B_{i}(l)|\rightarrow0$ . Indeed, an alternative description
for uniform reporting is a contagion on an underlining full mesh grid,
and in this case the two processes are identical.

\subsection*{Small Infections}

We now consider the setting of very small infections, and prove that
in these settings, one can take $K$ to be a constant ($K\in\Theta(1)$),
and moreover can choose this constant so that even the presence of
a single hotspot indicates an epidemic infection. 
\begin{thm}
\label{prop: small epidemics converges} Assume the total number of
reporting nodes, is $N^{1-\beta}$, while the number of truthfully
reporting nodes is $N^{\rho}$ and the truly reporting probability
is $q(N)=N^{-\mu}$. Set $K=\left\lceil 1/\beta\right\rceil -1$ .
If $\gamma(K)\in\omega\left(\log^{-1}(N)\right),$ then the hotspot
aggregator algorithm with parameters $K$ and threshold $T=1$ classifies
correctly with high probability if $K\mu\leq\rho.$ In particular,
if $1>\beta>0.5$, than the hotspot algorithm with $K=1$ classifies
correctly under the same conditions.
\end{thm}
In this last case, the algorithm simply counts the expected number
of infected neighboring pairs.
\begin{proof}
For each two nodes, $i$ and $j$, the hotspot indicators $H_{i}$
and $H_{j}$ are independent iff the corresponding environment $B_{i}$
and $B_{j}$ are disjoint, $B_{i}\cap B_{j}=\emptyset$, otherwise
they are positively correlated. In the uniform reporting scenario,
define the probability that all the hotspot indicators are false as
$P_{ep}\triangleq\Pr\left(\forall i\in N_{\mbox{{\tiny{reporting}}}},H_{i}=0\right)$.
The probability for such event is bounded by:

\begin{eqnarray*}
P_{ep} & \geq & \prod_{i=0}^{N_{\mbox{{\tiny{reporting}}}}-1}\Pr(H_{i}=0,H_{j}=0|B_{i}\cap B_{j}=\emptyset)\\
 & = & \left(1-p^{k}\right)^{N_{\mbox{{\tiny{reporting}}}}}.
\end{eqnarray*}

Therefore, as $p\rightarrow0$ ,$N_{\mbox{{\tiny{reporting}}}}\rightarrow\infty$

\[
\Pr\left(\forall i\in N_{\mbox{{\tiny{reporting}}}},H_{i}=0\right)\rightarrow\exp\left(-p^{k}N_{\mbox{{\tiny{reporting}}}}\right)
\]

According to the central limit theorem, $N_{\mbox{{\tiny{reporting}}}}=pN$
with high probability. By Applying Lemma \ref{lem:binmoial sum} and
choosing $k$ such that $p^{k}\in O(N^{-1})$, we have $\Pr\left(\forall i\in N_{\mbox{{\tiny{reporting}}}},H_{i}=0|\mathcal{\tilde{\mathcal{I}}}\right)\rightarrow1.$

Similarly to Theorem \ref{prop:non-vanishing boundary}, set $P_{in}=p_{in}^{k}$,
$N'_{\mbox{{\tiny{reporting}}}}=\gamma N{}_{\mbox{{\tiny{reporting}}}}/(f+1)$.
In the epidemic setting, we can again apply Corollary 2.6 in \citet{Janson2004},
but now with a deviation of $t=N'_{\mbox{{\tiny{reporting}}}}P$.

Then, the type II error is bounded by

\[
E_{II}\leq\exp\left(-\frac{N'_{\mbox{{\tiny{reporting}}}}P\left(1-\left(K^{2}+1\right)/4N\right)}{2\left(K^{2}+1\right)}\right).
\]

The type II error tends to zero if $N'_{\mbox{{\tiny{reporting}}}}P\rightarrow\infty.$
The latter condition can also be written as $q\alpha\gamma\left(q^{K}+\left(\alpha qf\right)^{K}\right)\in\omega\left(N^{-1}\right)$,
while the condition for the type II correct convergence is $\left(\alpha qf\right)^{K+1}\in O(N^{-1}).$
\end{proof}
In particular, if $N_{\mbox{{\tiny{reporting}}}}/N\in O(N^{-0.5})$,
and $q^{2}\alpha\in\omega\left(N^{-1}\right)$ then it is sufficient
to choose $K=1$, i.e., to count the number of infected nearest neighbors
pairs. 

This result is particularly useful for small infections (for example,
$\alpha\in\Theta\left(N^{-0.3}\right)$). It shows that it is possible
to detect epidemics using the hotspot aggregator algorithm even when
the ratio of the number of false positives to the number of infected
nodes tends to infinity (for example, $f\in\Theta\left(N^{0.7}\right)$),
and the ratio of truly reporting nodes number to the number of infected
nodes tends to infinity (for example, $q\in\Theta\left(N^{-0.3}\right)$),
a \emph{quadratic ``needle in a haystack''} scenario.

Finally, note that the requirement for the algorithm convergence is
that the boundary would not contain all the nodes. Therefore, even
if multiple sources of an epidemic exist, then as long as this condition
is satisfied, the algorithm converges correctly. Put differently,
the hotspots aggregator will converge correctly as long as there are
sources that had the chance to infect their \emph{local} environment,
rather than infect a large portion of the network. Hence, this algorithm
is able to identify and nip contagions in the bud.

Next, we deploy this algorithm on both random network models, such
as Erdos-Renyi networks or scale-free networks, and real-world networks,
such as an enterprise email network, the Internet AS topology and
Facebook.

\section{\label{sec:Experiments}Experiments}

We perform empirical tests of our algorithm on both synthetic and
real data. We focus on demonstrating its accuracy, running time and
scalability in numerous settings. These simulations also demonstrate
the ease-of-use of our algorithm in real experiments. While the theorems
themselves require some care in choosing the parameters of the algorithm,
here we find that 1, 2, and 3-hop local neighborhoods perform extremely
well across a wide range of settings, including the setting of an
overwhelming number of false positives, and the setting of up to 200
infection seeds (initial infected nodes).

\textbf{Synthetic Graphs}. We consider the algorithm performance in
the acute regime, where the number of false positives and false negatives
each tends to infinity. We tested this (Fig. \ref{fig:error rate})
on an Erdos-Renyi graph in the regime where the giant component emerges,
$G(N=8000,p=2/N$) (\citet{Durrett2010}). The error rate decays rapidly
as the size of the graph increases, as predicted by Theorem \ref{prop: small epidemics converges},
even as the fraction of truly reporting nodes among the reporting
nodes tends to zero. Moreover there are infinitely more false negatives
than true positives, as the network size increases. We also compare
against the Median-ball-algorithm given in \citet{Milling2013}. While
the Median algorithm may be converging, we see that the convergence
of our algorithms for $l=1,2,3,4$ is remarkably faster.

Interestingly, in this experiment, the $l=4$ version of our algorithms
seems to have a slightly worse performance compared to $l=1,2$. This
is likely because the regime of $G(N,p)$ we consider has diameter
$\Theta\left(\log N\right)$, and thus even 4-hop neighborhoods are
considerable in size.

In real-world networks, the degree distribution often follows a power
law. For random power law networks, the network diameter is, like
in the $G(N,p)$ , $\Theta(\log N)$, or even $\Theta(\log\left(\log N\right))$
(\citet{Cohen2003}). In such \emph{small-world} networks, choosing
a large ball radius deteriorates inference performance due to an overflow
of the local ball about an infected node to an uncontaminated region.

\begin{figure}
\centering{}\centering\includegraphics[width=0.8\columnwidth]{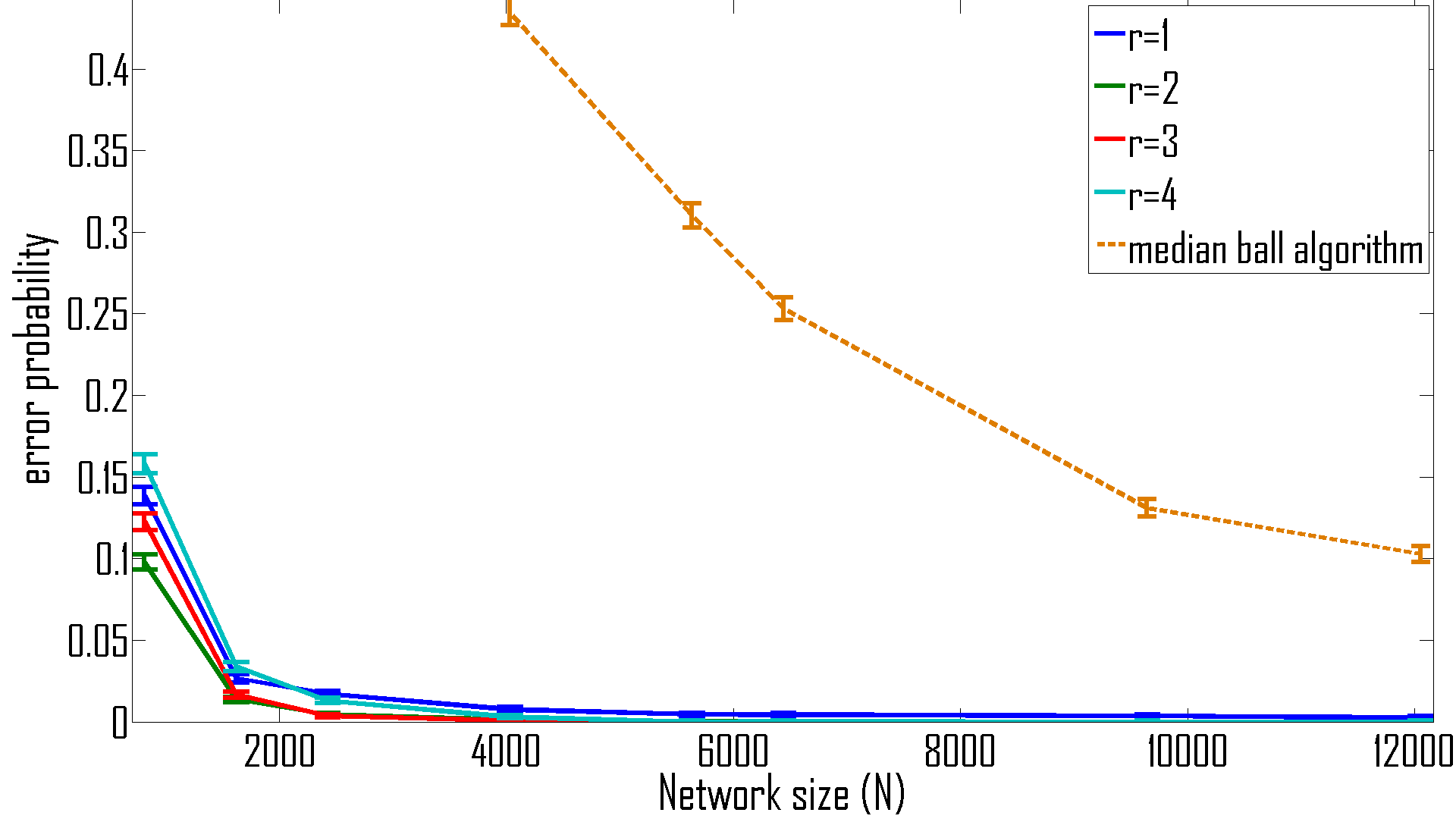}\caption{\label{fig:error rate}The mean error rate for different ball radii
for an Erdos Renyi graph $G(N,p=2/N)$. In this configuration, the
number of false reporting nodes tends to infinity, while the fraction
of truly reporting nodes among the reporting nodes tends to zero.
In addition, the true reporting probability tends to zero, so that
the number of false negative is infinitely higher than the number
of true positives. Furthermore, the epidemic size is small, $O(N^{0.3})$.
The simulation details are elaborated in the appendix.}
\end{figure}

\textbf{Real World Graphs}. We next consider a real world graph, namely,
the enterprise email network of Enron employees (\citet{Klimt}).
In part, this is motivated by the fact that many computer viruses
spread by email attachments. Figure \ref{fig:enron} shows the setting
of a small infection, with a very large fraction of false negatives,
a large majority of false positives, \emph{and many initial sources
of infection}. Our results demonstrate the stability of our algorithm
with respect to the number of initial infection seeds. Even with as
many as 200 initial seeds, the performance of our algorithm is hardly
affected. We note again a better performance for radius values $r=1,2,3$
in our algorithm; this is for the same reason as discussed above.

Additional instances of diffusive processes are the sharing of viral
media and the adoption of memes on social networks. We examined the
identification of such phenomena by testing our algorithm on a network
of $63K$ Facebook contacts (\citet{Viswanath2009}). Finally, cascades
of router failures due to misconfiguration or BGP attacks is a major
security concern. We have performed experiments on these graphs as
well, but for space concerns defer these to the appendix.

\begin{figure}
\centering{}\centering\includegraphics[width=0.75\columnwidth]{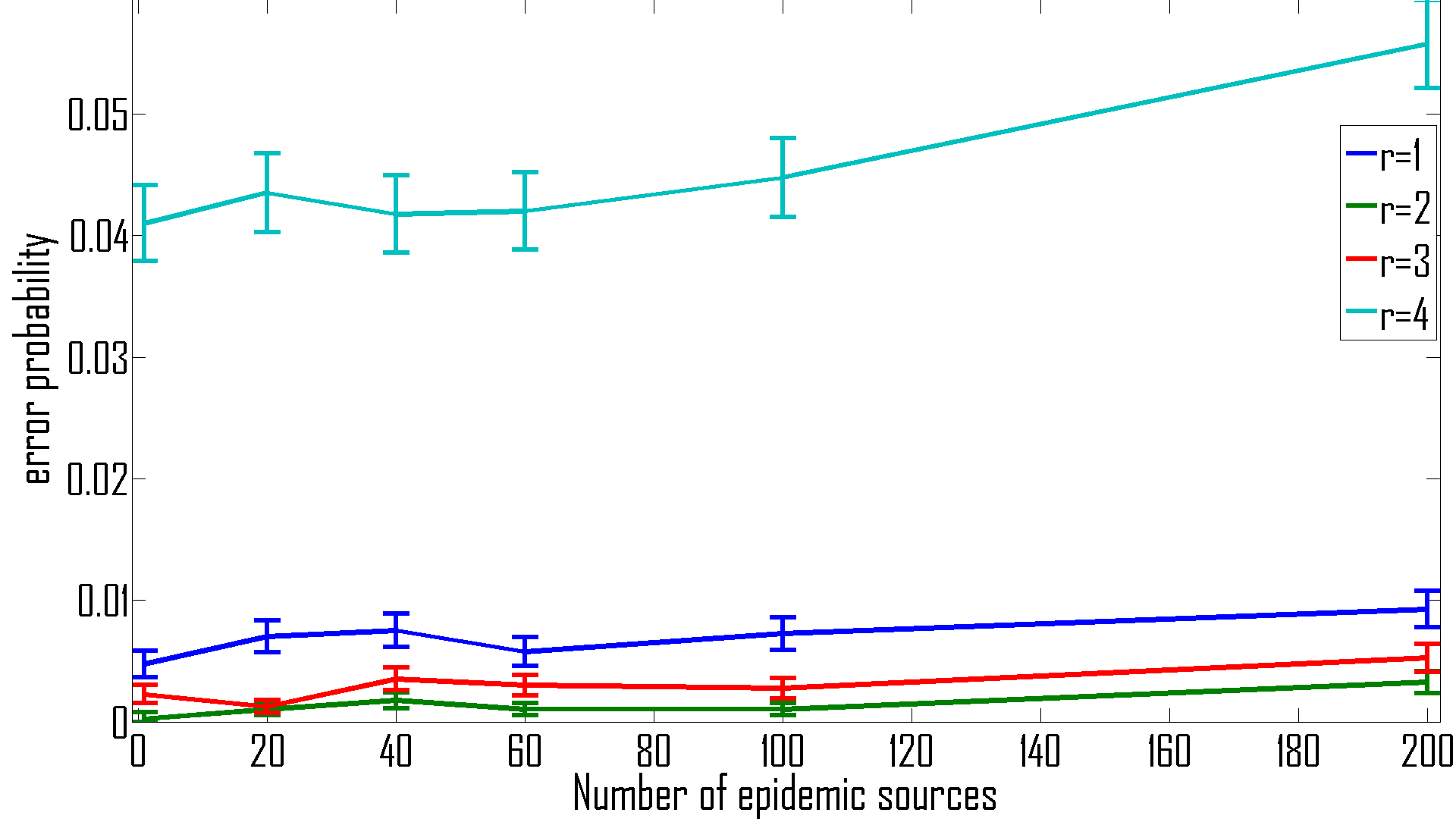}\caption{\label{fig:enron}The mean error rate as a function of the number
of epidemic sources for Enron email network. The network is comprised
of $N=33,696$ nodes. The number of infected nodes is small, $0.06N$,
and the number of truly reporting nodes is miniscule, $0.003N$. Finally,
for each truly reporting node there are $8$ false positives ($f=8$).
Even in this challenging situation, the error rate is low. }
\end{figure}

Last, we tested our algorithm in the presence of \emph{noisy network}
information, to test robustness to knowledge of local neighborhoods.
Indeed, a person may erroneously estimate the distance to her peers,
consequently resulting in an incorrect set of nearest neighbors, or
the deformation of the ball centered about her. In our experiments,
we account for the effect of noisy network information by allowing
an expected fraction of the conceived inter-node distances to increase
or decrease by $d$ (when the unmodified distance is greater than
$d$). Consequently, the inference algorithm is performed on a different
network than the one the epidemic evolved on. Note that such a modification
is not symmetric, as one node may correctly estimate the distance
to its peer, while its peer may not. The error rate of our inference
algorithm is present in Fig. \ref{fig:distance error}. The results
here show that our algorithm is robust to such noisy network knowledge.

In an epidemic scenario, if a reporting node is deep in the infected
region, there exists some buffer to the uninfected nodes at the contagion
border, and the overestimation effect should be negligible. In contrast,
if the reporting node's exterior shell is close to the border, then
the reporting node may mistake an uninfected node as a member of its
local ball and performance may deteriorate. Therefore, we expect that
the type II error should increase with the ball radius and the misestimation
rate, expressed in either increased misestimation probability or the
true distance deviation. These two effects are observed in Fig. \ref{fig:distance error}.
Finally, the type I error is unaffected by this type of topological
noise.

\begin{figure}
\begin{centering}
\includegraphics[width=0.8\columnwidth]{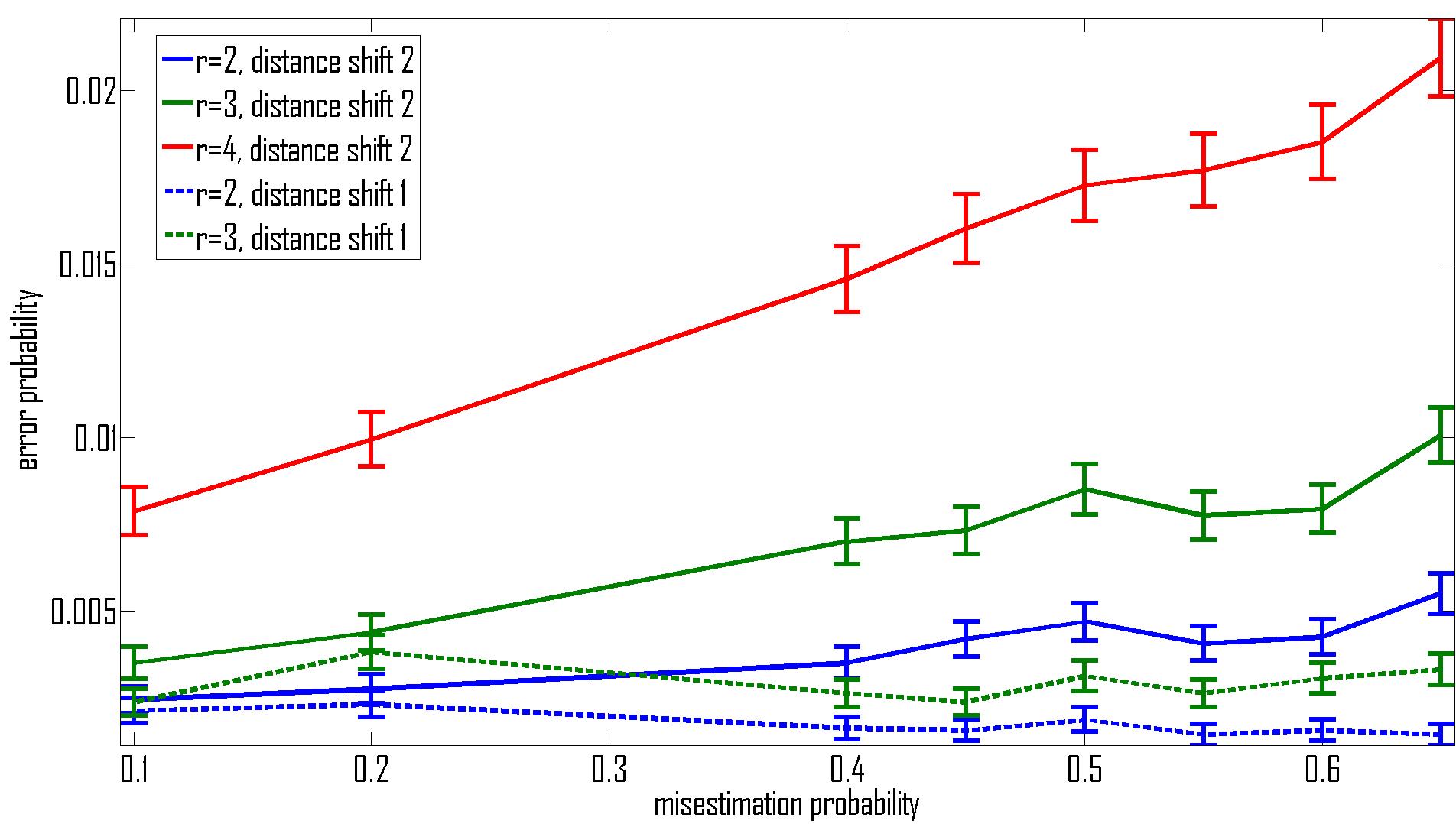} 
\par\end{centering}

\caption{\label{fig:distance error}The error probability under additional
topological noise for an Erdos Renyi network $G(N=4000,p=2/N)$. The
mean error is plotted against the probability that a node will estimate
its distance to another node. The fault is equally likely to be an
overestimate or underestimate of one (dashed lines) or two (solid
lines) from the true nodes distance. The number of infected nodes
is about $0.2$ of the whole population of $N=4000$. The number of
true positives is $0.06N$, while the number of false positives is
$0.08N$.}
\end{figure}

\section{Conclusions}

Our world becomes increasingly interconnected by multiple different
networks. While this interconnection speeds information transfer and
facilitates communication, is it also increases the likelihood of
contagion outbreaks. Depending on the context, such outbreaks may
be computer viruses, opinion trends, or epidemics such as malaria.
The hair-trigger identification of these diffusive processes is crucial
and one may be forced to rely on a single reporting map snapshot in
time. Furthermore, a complete knowledge of the underlying network
graph is generally unrealistic. We have provided both analytical proofs
and experiments on real data, showing that the hotspot aggregator
algorithm solve this statistical inference problem in its dire setting
on a wide variety of network and settings. We have shown that this
algorithm classifies correctly even when the rate of false positives
and false negatives is infinitely higher than the number of true positives,
as well as in the presence of multiple sources of epidemics.

We have simulated this algorithm both on random networks, such as
Erdos-Renyi graphs or scale-free graphs, and on real world networks,
such as Facebook, the Internet AS-topology and Enron email chain.
Our simulations exhibit the exponential decrease of the error rate
with the size of the network, as predicted. In practice, the error
rates are extremely low, even when there are errors in estimating
the network structure. Finally, the complexity of our algorithm is
low, and it is clearly scalable, as shown by the low error rate on
the Facebook network of over $60K$ nodes.

\pagebreak{}

\raggedbottom

\bibliographystyle{icml2014}

\pagebreak{}

\appendix
%dummy comment inserted by tex2lyx to ensure that this paragraph is not empty

\section{Appendix A - Additional lemmas}
\begin{lem}
\label{lem: inside outside random probabilities} Denote the reporting
probability of an infected node in the epidemic scenario by $p_{in}$.
Similarly, denote the node reporting probability in the uniform reporting
scenario as $p$. Then, $p_{in}>p.$\end{lem}
\begin{proof}
First, note that 
\begin{eqnarray*}
p_{in} & = & q+(1-q)\frac{fq\alpha}{1-q\alpha}\\
 & = & q\left(1+(1-q)\frac{f\alpha}{1-q\alpha}\right)
\end{eqnarray*}

while 
\[
p=\left(f+1\right)q\alpha
\]

Therefore,

\begin{eqnarray*}
p_{in}-p & > & q+(1-q)\frac{fq\alpha}{1-q\alpha}-\left(f+1\right)q\alpha\\
 & = & \frac{q}{1-q\alpha}\left(1-q\alpha+(1-q)f\alpha-\left(f+1\right)\alpha\left(1-q\alpha\right)\right)\\
 & = & \frac{q}{1-q\alpha}\left(1-q\alpha-qf\alpha+f\alpha-f\alpha-\alpha+q\alpha^{2}+fq\alpha^{2}\right)\\
 & = & \frac{q}{1-q\alpha}\left(\left(1-\alpha\right)(1+q\alpha+qf\alpha\right)\\
 & > & 0
\end{eqnarray*}

as $\alpha>0$.\end{proof}
\begin{lem}
\label{lem:The-errors-expressed}The errors expressed in Theorem \ref{prop:non-vanishing boundary}
tends to zero under the corresponding conditions.\end{lem}
\begin{proof}
Set $P=p^{K},P_{in}=\gamma p_{in}^{k}/(f+1)$. According to Corollary
2.6 in \citet{Janson2004} the type I error rate can be bounded by

\[
E_{I}\leq\exp\left(-N_{\mbox{{\tiny{reporting}}}}\frac{\left(P-P_{in}\right)^{2}\left(1-\frac{\left(\Delta(\Gamma)+1\right)}{4N_{\mbox{{\tiny{reporting}}}}}\right)}{8\left(\Delta(\Gamma)+1\right)\left(1-P\right)\left(P+\left(P_{in}-P\right)/6\right)}\right)
\]

where $\Delta(\Gamma)$ is the maximal degree of the dependency graph
$\Gamma.$ However, $\Delta(\Gamma)\leq\max\{N_{\mbox{{\tiny{reporting}}}},K^{2}\}$
and therefore, 
\begin{equation}
E_{I}\leq\exp\left(-N_{\mbox{{\tiny{reporting}}}}\frac{\left(P_{in}-P\right)^{2}}{16\left(K^{2}+1\right)\left(P+\left(P_{in}-P\right)/6\right)}\right)\label{eq:E1}
\end{equation}

Set $\delta=p_{in}-p.$ First, note that if $\lim\left(P-P_{in}\right)\neq0$,
then this expression tends to zero. Therefore, if $\lim\delta\neq0$
for $N\rightarrow\infty$ than the type I error rate tends to zero.

Recall that

\begin{eqnarray*}
p_{in} & \in & \Omega\left(q+fq\alpha\right)\\
\delta & \in & \Theta(q)\\
p & \in & \Theta\left(fq\alpha\right)
\end{eqnarray*}

And therefore in this case $q\in\Omega(1)$. Otherwise, assume that
$\lim\delta=0$. If 
\[
\delta/p\rightarrow0
\]

or equivalently, $f\alpha\in\omega(1)$, the highly noise regime,
then to first order in $\delta/p$

\[
\left(P_{in}-P\right)^{2}\rightarrow p^{2k}\left(1-\frac{1}{\gamma(f+1)}\right)^{2}
\]

with $\delta=p_{in}-p.$ Likewise, if $\delta/p\rightarrow const$,
than we can write 
\begin{eqnarray*}
p & = & g(N)\tilde{p}\\
\delta & = & g(N)\tilde{\delta}
\end{eqnarray*}

where $\tilde{p}\rightarrow const\neq0$ and $\tilde{\delta}\rightarrow const\neq0$
and 
\[
\left(P-P_{in}\right)^{2}=g^{2K}(N)\left(\frac{\left(\tilde{p}+\tilde{\delta}\right)^{K}}{\gamma(f+1)}-\tilde{p}^{K}\right)^{2}
\]
Therefore, for either $K=1$ or $K=2$ we have 
\[
\lim\left(\frac{\left(\tilde{p}+\tilde{\delta}\right)^{K}}{\gamma(f+1)}-\tilde{p}^{K}\right)^{2}\neq0
\]

and therefore $\left(P-P_{in}\right)^{2}\in\Theta(p^{2k})=\Theta(\delta^{2k})$.

Recall that with high probability $N_{\mbox{{\tiny{reporting}}}}=Np$.
As, 
\[
E_{I}\leq\exp\left(-N\Theta\left(p^{2k+1}\right)\right)
\]

if 
\[
\left(fq\alpha\right)^{2K+1}\in\omega(N^{-1}),
\]

the type I error tend to 0 for $N\rightarrow\infty$. In other words,
if $K$ is such that 
\[
\left(N_{\mbox{{\tiny{reporting}}}}/N\right)^{2K+1}\in\omega(N^{-1})
\]

the error tend to 0. Alternatively, if $p/\delta\rightarrow0$ then

\[
\left(P_{in}-P\right)^{2}\rightarrow\left(\frac{\gamma\delta^{k}}{(f+1)}\right)^{2}
\]

and therefore, if 
\[
\gamma^{2}q^{K}/f\in\omega(N^{-0.5})
\]

then the algorithm converges correctly.

A similar calculation shows that if this condition holds, then the
type II error tends to 0 for $N\rightarrow\infty$. Explicitly:

\[
E_{II}\leq\exp\left(-N_{\mbox{{\tiny{reporting}}}}\frac{\left(P-P_{in}\right)^{2}\left(1-\frac{\left(\Delta(\Gamma)+1\right)}{4N_{\mbox{{\tiny{reporting}}}}}\right)}{8\left(\Delta(\Gamma)+1\right)P}\right)
\]

and following similar reasoning 
\[
E_{I}\leq\exp\left(-N_{\mbox{{\tiny{reporting}}}}\frac{\left(P_{in}-P\right)^{2}}{16\left(K^{2}+1\right)P_{in}}\right)
\]

and this expression has the same scaling properties as eq. \ref{eq:E1}.\end{proof}
\begin{cor}
Consider a large infection, where the number of infected nodes is
a constant fraction of the number of nodes, $\alpha=\mbox{\ensuremath{\Theta}(1) .}$
The algorithm converge correctly even when the reporting probability
tends to zero as $q\in\omega\left(N^{-1}\right)$ while the noise
ratio tend to infinity $f\in\omega\left(q^{-1}\right)$, as long as
conditions a) and b) of theorem \ref{prop:non-vanishing boundary}
are satisified.

In particular, if the number of truly reporting nodes is $\Theta(N)$,
then the algorithm converges correctly if $\gamma(K=const)>0$. Alternatively,
the algorithm converges correctly for every network such that $\gamma(K=2\log N)>0$,
for example, grids and tree like networks, even if the noise is $f\in\Theta(N)$. \end{cor}
\begin{proof}
Note that $f<N$. If $\gamma\neq0,$ then $\gamma<1/\alpha N$. In
particular, if the number of truly reporting nodes is $\Theta(N)$,
then $f\in\Theta(1)$.Therefore 
\[
K\geq-\log(\gamma)+\log(f)
\]

and substituting the corresponding values in Theorem \ref{prop:non-vanishing boundary}
concludes the proof..\end{proof}
\begin{lem}
\label{lem:binmoial sum}Consider a binomial random variable $X$
with expectation value $pN$ for $N\rightarrow\infty$. Then any sum
of Bernoulli random variables $Y_{i}$ obeys 
\begin{eqnarray*}
\mathbb{E}\left(\sum_{1}^{X}Y_{i}\geq c\right) & \rightarrow & \mathbb{E}\left(\sum_{1}^{M}Y_{i}\geq c\right)\\
\mathbb{E}\left(\sum_{1}^{X}Y_{i}\leq c\right) & \rightarrow & \mathbb{E}\left(\sum_{1}^{M}Y_{i}\leq c\right)
\end{eqnarray*}

In other words, we can asymptotically replace the summation over random
number of RVs in the summation over the mean number or RVs.\end{lem}
\begin{proof}
Set $M=Np-\left(Np\right)^{2/3}$. First, note that

\begin{eqnarray*}
\mathbb{E}\left(\sum_{1}^{X}Y\geq c\right) & = & \mathbb{E}\left(\sum_{1}^{X}Y\geq c|X>M\right)\Pr(X>M)+\mathbb{E}\left(\sum_{1}^{X}Y\geq c|X\leq M\right)\Pr(X\leq M)\\
 & \leq & \mathbb{E}\left(\sum_{1}^{X}Y\geq c|X>M\right)\Pr(X>M)\\
 & \leq & \mathbb{E}\left(\sum_{1}^{X}Y\geq c|X=M\right)\Pr(X>M).
\end{eqnarray*}

According to Hoeffding's inequality 
\begin{eqnarray*}
\Pr(X\leq M) & \leq & \exp\left(-\left(Np\right)^{1/3}\right).
\end{eqnarray*}

and therefore 
\[
\mathbb{E}\left(\sum_{1}^{X}Y\geq c\right)\leq\mathbb{E}\left(\sum_{1}^{M}Y\geq c\right)\left(1-\exp\left(-\left(Np\right)^{1/3}\right)\right)
\]

Similarly, we can apply Hoeffding's inequality for $M=Np-\left(Np\right)^{2/3}$
in

\begin{eqnarray*}
\mathbb{E}\left(\sum_{1}^{X}Y\leq c\right) & = & \mathbb{E}\left(\sum_{1}^{X}Y\leq c|X>M'\right)\Pr(X>M)+\mathbb{E}\left(\sum_{1}^{M}Y\leq c|X\leq M'\right)\Pr(X\leq M)\\
 & \leq & \mathbb{E}\left(\sum_{1}^{X}Y\geq c|X\leq M'\right)\Pr(X<M)
\end{eqnarray*}

while 
\[
\Pr(X\geq M')\leq\exp\left(-\left(Np\right)^{1/3}\right)
\]

and obtain

\[
\mathbb{E}\left(\sum_{1}^{X}Y\leq c\right)\leq\mathbb{E}\left(\sum_{1}^{M'}Y\geq c\right)\left(1-\exp\left(-\left(Np\right)^{1/3}\right)\right)
\]

but both $M,M'\rightarrow Np$ and this concludes the proof. 
\end{proof}

\section{Appendix B - Algorithm application}

In order to apply the algorithm, the parameters $K$ and $T$ must
be specified, and the number of nodes in the local environment must
be set. We assume that a good estimate for the reporting probability
of an infected node, $p_{in}$, is known. Note that for large networks
the reporting probability $p$ can be easily estimated according to
$p=N_{\mbox{{\tiny{reporting}}}}/N$. Recall that $\gamma(K)$ is
the fraction of infected nodes for which their $K$ nearest neighbor
nodes are also infected is known.

We distinguish between two classes of networks. First we discuss large
networks for which the function $\gamma(K)$ is known, whether numerically
or analytically. As an example, consider an infinite constant degree
tree with degree $d$. Assume that the infection, starting from the
root, has infected all the nodes up to the $m$-th level. In this
case, for every node that is deeper than the $m-l$ level of the tree,
all the nearest $d^{l}$ are not infected. The fraction of such nodes
is
\[
\gamma(K=d^{l})\leq\frac{d^{m-l}-1}{d^{m}}\rightarrow d^{-l}
\]

Therefore, $\gamma(K)=K^{-1}$. 

As another example, consider a $d$ dimensional grid. In this case,
the contagion is with high probability contained within an $l_{1}$
ball of radius $r=c|S|^{1/d}$, with $c\rightarrow const$ for $N\rightarrow\infty$
(\citet{Milling2012}). The number of nodes on the non-infected bordering
nodes near the surface of such ball is $\Theta(r^{d-1})$. In addition
a ball of radius $l$ contains $\Theta(l^{d})$ nodes. Therefore,
the number of interior nodes is at least $\Theta(r^{d-1}l^{d})$.
Hence, 
\[
\gamma\left(K\in\Theta(l^{d})\right)\leq\Theta\left(r^{-1}l^{d}\right)
\]

Hence for balls of radii $K\in o\left(|S|^{1/d}\right)=o\left(\left(\alpha N\right){}^{1/d}\right)$
we have 
\[
\gamma(K)\rightarrow1
\]

In these cases, Theorem \ref{prop:non-vanishing boundary} provides
a prescription for $K$ and $T$ when the number of reporting nodes
is large, $\Theta(N)$, while if the number of reporting nodes is
small, $o(N)$, the prescription in Theorem \ref{prop: small epidemics converges}
applies

For a finite size network, that is not necessarily sampled from a
statistical ensemble, $\gamma(K)$ might be not known apriori. Nevertheless,$\gamma(K)$
may be calculated at a preprocessing stage. This can be done by simulating
epidemics according to the scenario details. Then, for each infected
node, calculate the probability that $K$ of its neighboring nodes
are infected, and average over all infected nodes. As an example,
the $\gamma(K)$ function for the Internet inter-AS topology, which
is often considered a scale free network, is presented in Fig. \ref{fig:gamma}.
Note that for any network, and any value of $K,$ if there is an infected
node that $K$ of its nearest neighbors are infected, then $\gamma(K)\neq0$.
In this case, $\gamma\geq1/|S|$, where$|S|$ is the infection size.
The maximal relevant value of $K$ is $\log(|S|)$, according to Theorem
\ref{prop:non-vanishing boundary}. 

After this preprocessing stage, if the number of reporting nodes is
large one applies Theorem \ref{prop:non-vanishing boundary} , or
Theorem \ref{prop: small epidemics converges} otherwise.

\begin{figure}
\includegraphics[width=1\columnwidth]{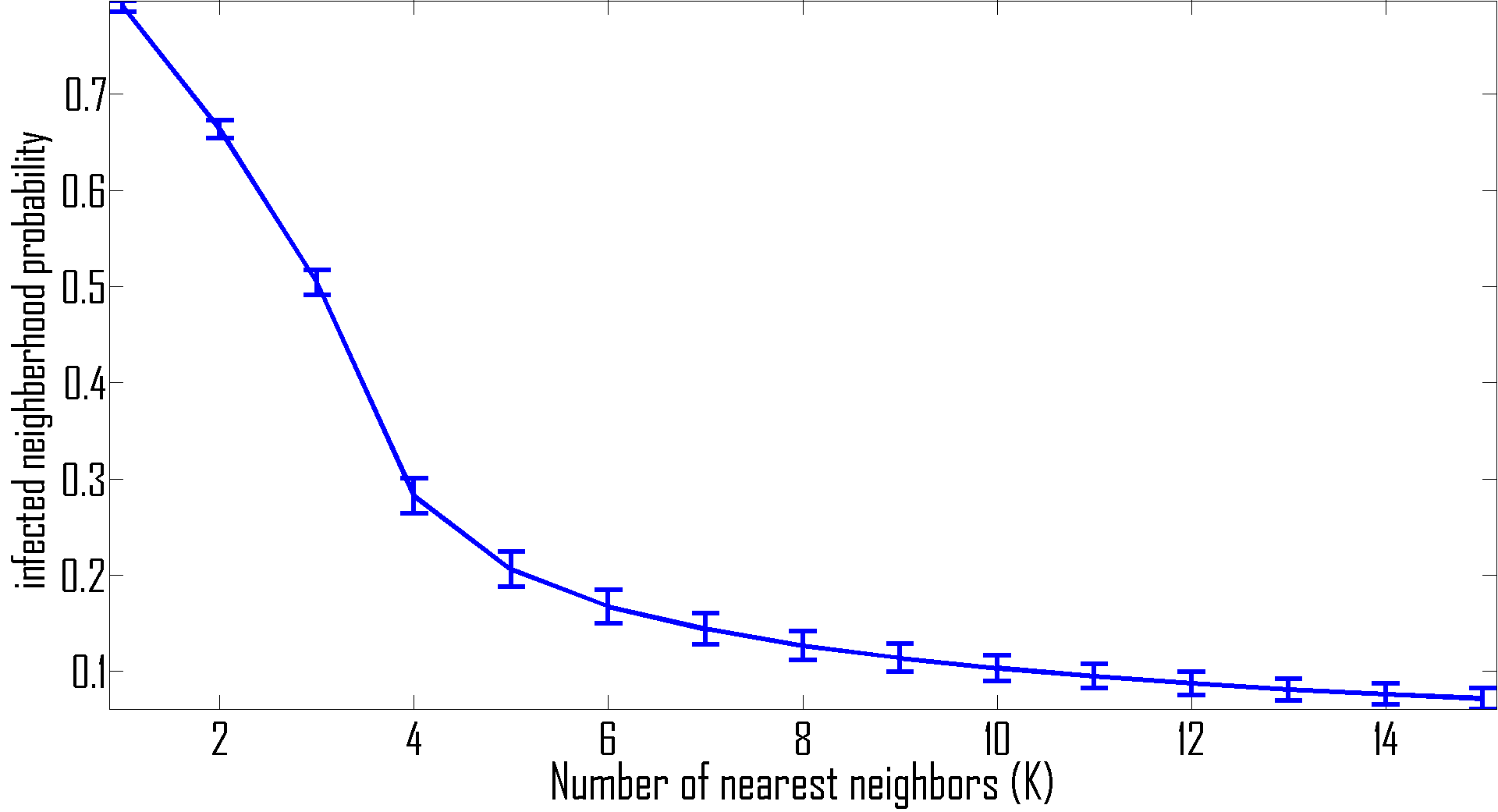}\caption{\label{fig:gamma}The probability $\gamma(K)$ that $K$ nearest neighbors
of an infected node will be infected as well. The plot was generated
for the Internet inter-AS graph, composed of $N\approx27K$ nodes.
Even for a small infection size $(|S|=0.06N$) the fraction $\gamma(K)$
is significant even up to the maximal relevant value of $K=\log(|S|)\approx7.5.$}
\end{figure}

Finally, it is important to note that one can perform multiple tests
with various values of $K$. If no epidemic occurred, then in all
those tests the number of hotspots should be close to the expected
number of a uniform reporting event. If any tests results is far from
this value, then an epidemic occurred. This approach is particularly
relevant if the critical requirement is high specificity, and if no
information on the possible epidemic is available. Note that as there
are at most $\log(N)$ tests, applying multiple tests does increase
the algorithm complexity appreciably.

\section{Appendix C - Additional experiments}

\begin{figure}
\centering{}\centering\includegraphics[width=1\columnwidth]{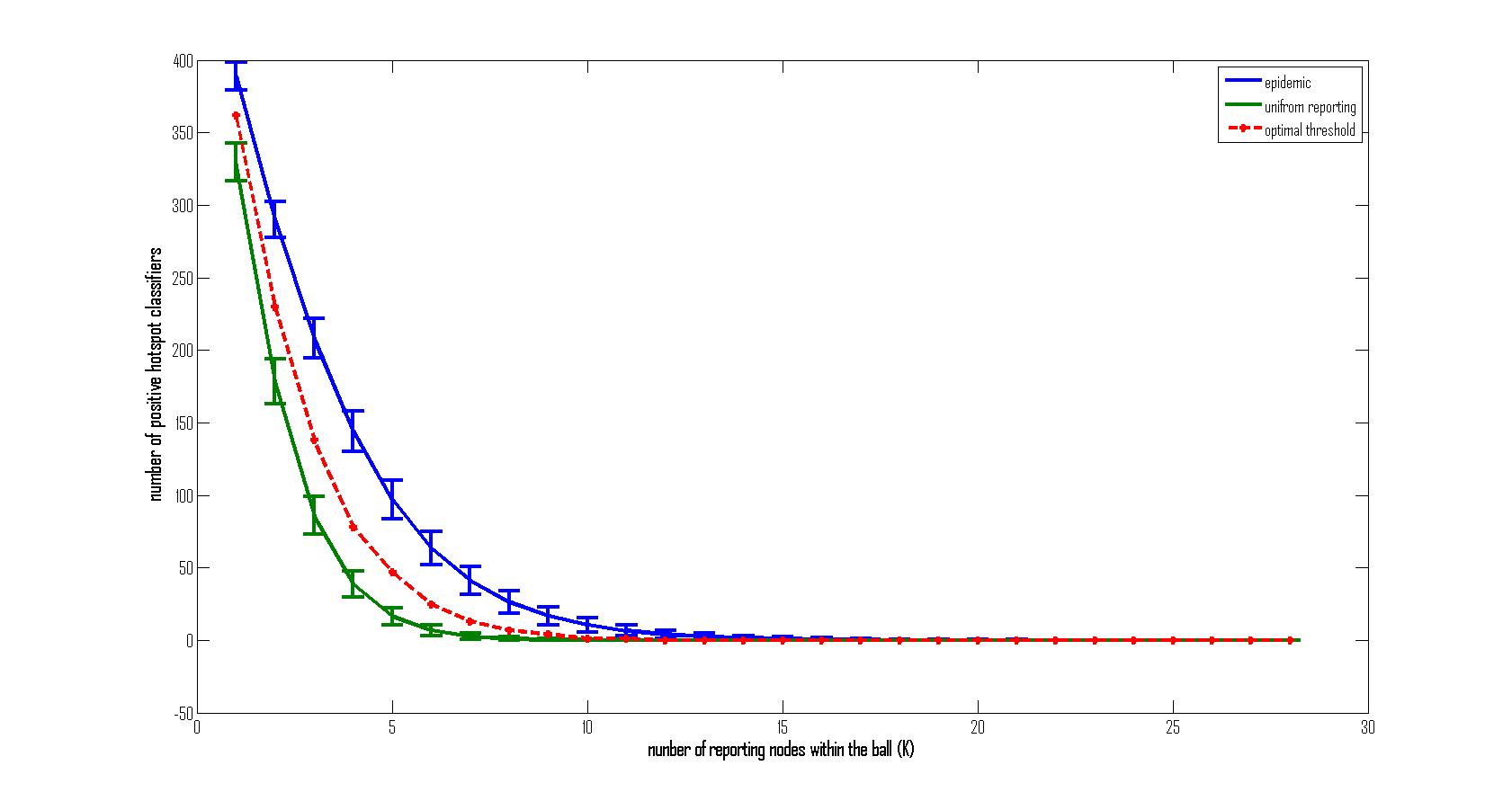}\caption{\label{fig:separibility}The mean number of hotspots (and the standard
deviation) as a function of $K$. The plot was generated for an Erdos
Renyi graph $G(N=8000,p=2/N)$ and ball radius $l=3$. The simulation
parameters are $\alpha=0.13,\: q=0.22$ and $f=1$. The plot was generated
by Monte Carlo simulations of $2000$ trials for each scenario.}
\end{figure}

In this section we present additional simulations, some of which are
described in the main body of this paper.

First, we considered idenfitiability on an Erd�s-Renyi graph in the
regime where the giant component emerges, $G(N=8000,p=2/N$) (\citet{Durrett2010}).
We performed $2000$ trials for each scenario (epidemic or uniform
reporting) and evaluated the mean number of hotspots as a function
of the threshold $K$ for a ball of a given radius. The resulting
graph for a ball of radius $l=3$ is presented in Fig. \ref{fig:separibility},
including the corresponding standard deviations. The graph shows that,
as our theory predicts, the number of hotspots is indeed a strong
statistic for distinguishing between an epidemic and a random illness.
This was followed, as described in the paper, by evaluating the algorithm
performance in the acute regime, where the number of false positives
and false negatives each tends to infinity (Fig. \ref{fig:error rate}).
In this figure, the infection size is $O(N^{0.3})$ while the reporting
probability is $O(N^{-0.2}$). The ratio of false positives to true
positives is $O(N^{0.3})$. The mean error is the average of the type
I and type II error probabilities, assuming equal likelihoods of an
epidemic and a uniform reporting event. The plot was obtained by Monte
Carlo simulations of $2000$ trials for each scenario.

We have followed suit and applied the algorithm on real world networks.
In Fig. \ref{fig:facebook graph} the algorithm's error rate on a
subgraph of the Facebook network is presented. Even though the noise
level is extremely high and the infection is tiny, the error rate
is negligible. We have also simulated our algorithm on the Internet
autonomous system (AS) network, in order to examine whether this algorithm
is able to detect failure cascades of BGP routers (Fig. \ref{fig:internet}).
In both cases, as discussed in Section \ref{sec:Experiments}, the
algorithm showed very good results, even for the large ball radius
$l=4$. These results were obtained in a challenging setting, under
a highly noisy environment with multiple epidemic sources, where only
a mere fraction of the nodes truly report. In both cases, the Median
Ball algorithm's error was close to $0.5$, almost as high as a random
classifier. While in principle this algorithm might be modified to
include multiple seeds, this modification requires knowledge of seed
number, which is rarely known. In addition, if the Median Ball algorithm
are modified, the type I error increases appreciably. 

\begin{figure}
\centering{}\includegraphics[width=1\columnwidth]{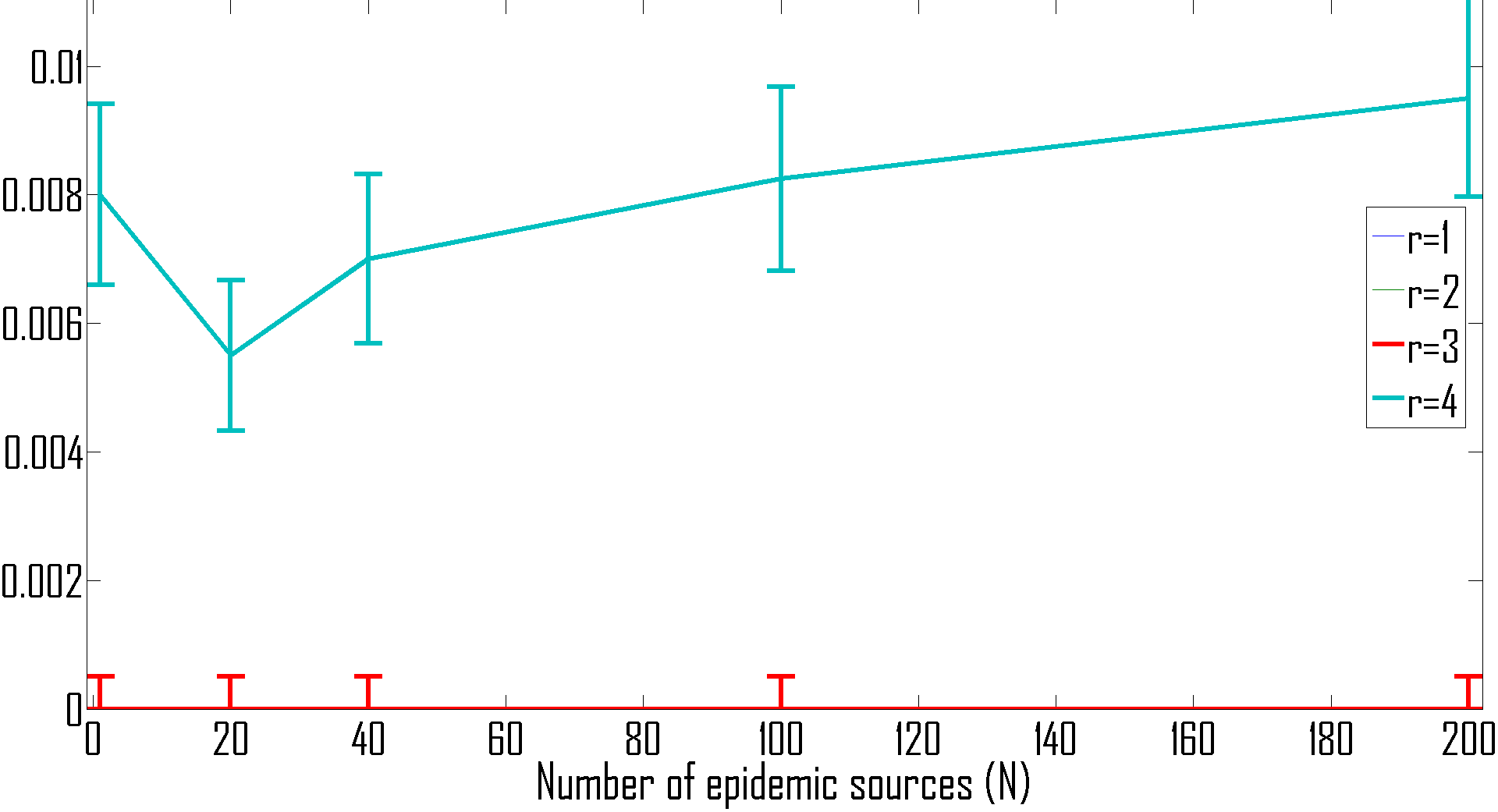}\caption{\label{fig:facebook graph}The error probability as a function of
the number of epidemic sources on a Facebook graph. The algorithm
with ball radii $r=1,2,3$ correctly classified every one of the $4000$
samples, and the error rate for the ball with radius $r=4$ is low
as well. The infected component is composed of only $3\%$ of the
entire graph. The number of truly reporting nodes is a mere $0.3\%$
of the graph. Finally, for each truly reporting node there are $8$
false positives ($f=8$). }
\end{figure}

\begin{figure}
\centering{}\includegraphics[width=1\columnwidth]{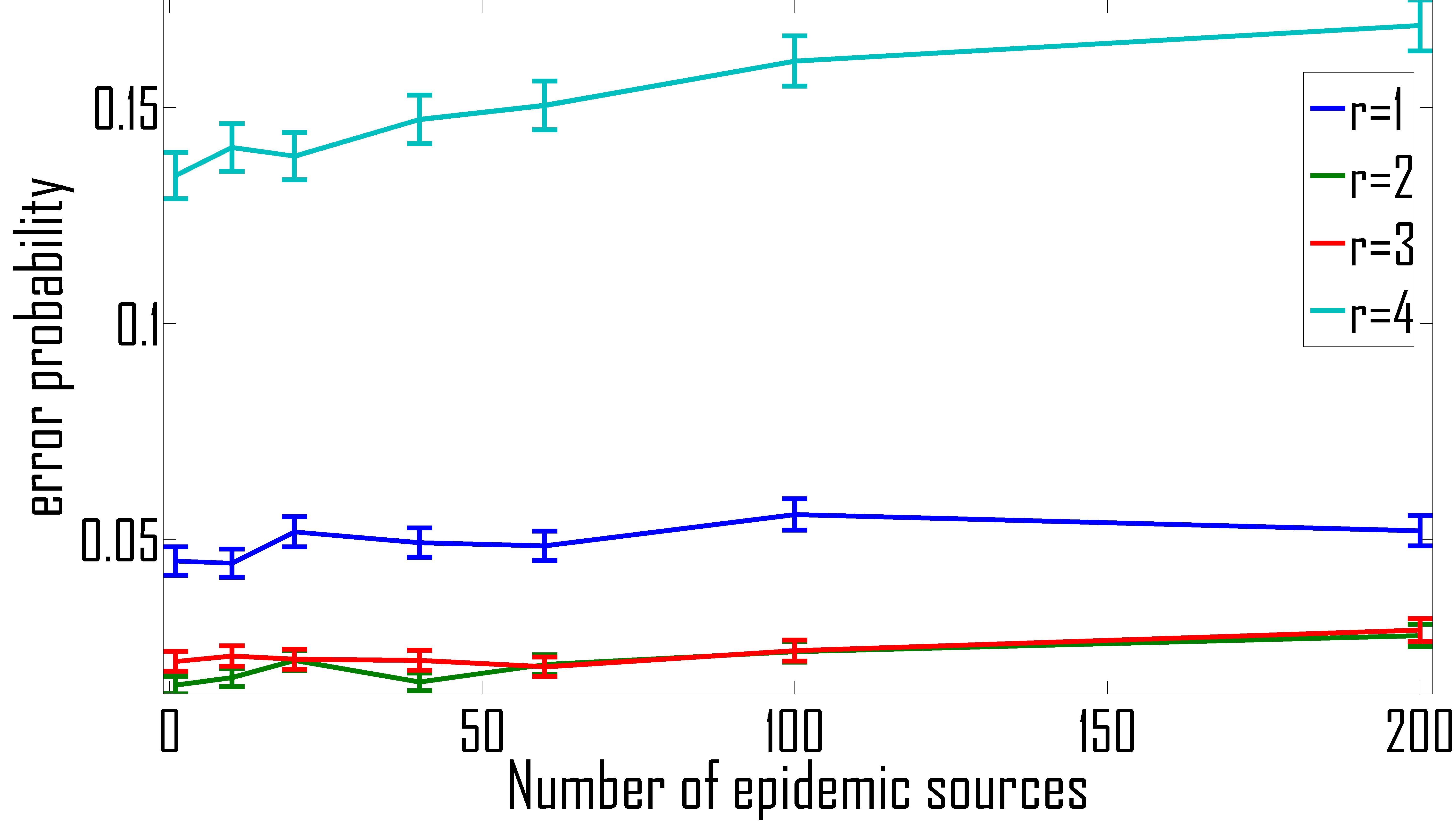}\caption{\label{fig:internet} The error probability as a function of number
of epidemic sources on the Internet AS graph. The Internet is well
known for its ultra small world property, and the mean distance between
nodes is less than four. The best results are achieved using balls
of either $r=2$ or $r=3$ radius. The number of infected nodes is
$6.7\%$ of the $27K$ ASs composing the internet. The reporting probability
is $q=0.05$ so only $0.3\%$ of the nodes in the graph are truly
reporting. Finally, for each truly reporting node there are $8$ false
positives ($f=8$). }
\end{figure}

The degree distribution of these networks closely resembles power
law. In particular, the degree distribution of most of this networks
is $Ax^{\alpha}$ with $3>\alpha>2$. We have tested our algorithm
on random power law (Fig. \ref{fig:scale free noise}). As the number
of reporting nodes increases, there are more positive classifiers
in the uniform reporting scenario. In order to compensate for this
effect, one needs to design a classifier for less frequent events.
This is done by increasing the threshold $K$, which in turn increases
the rarity of the event. Indeed, it is possible to achieve low error
rates with small ball radii even in the presence of a large number
of false positives.

\begin{figure}
\includegraphics[width=1\columnwidth]{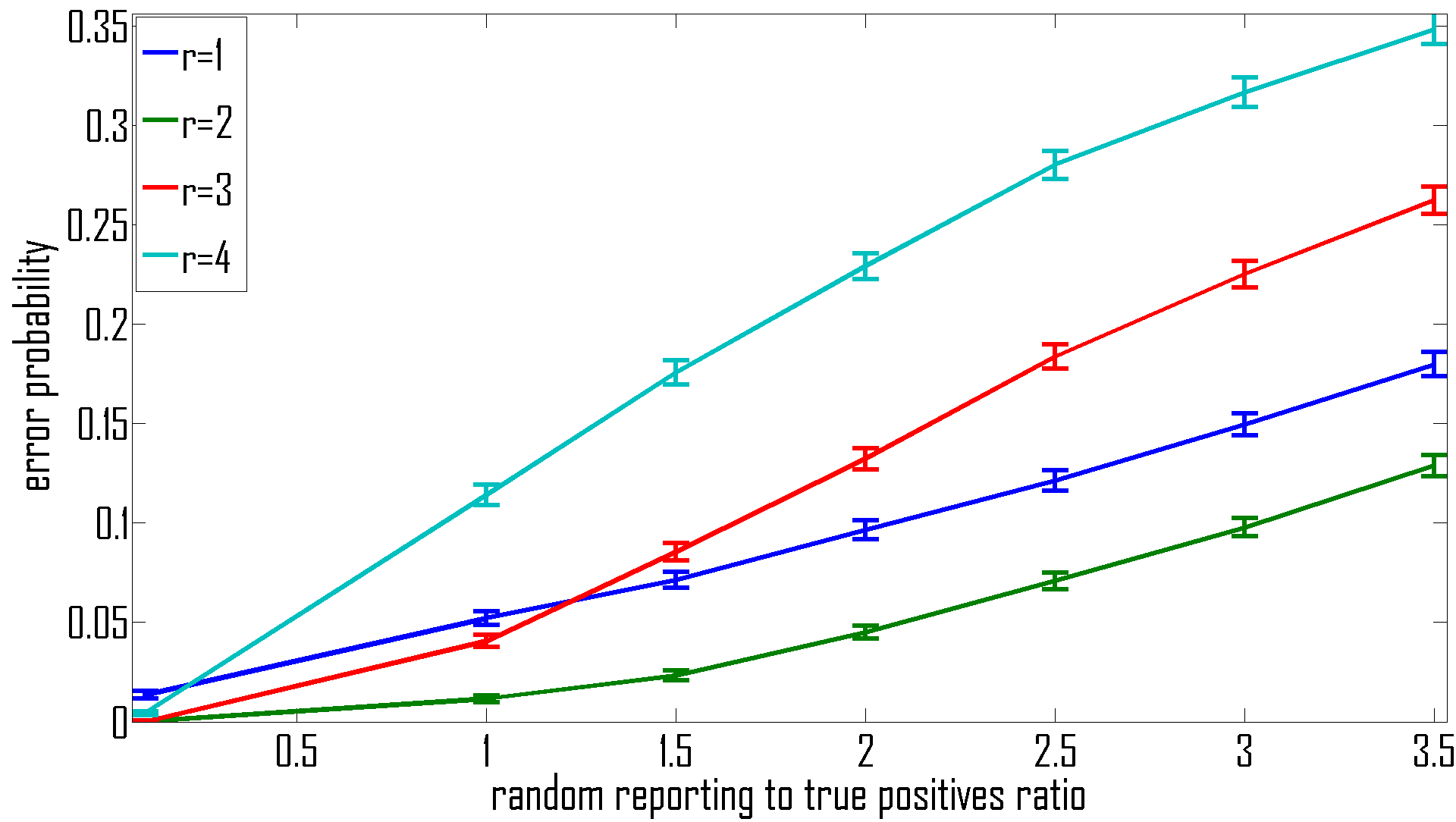}\caption{\label{fig:scale free noise}The error probability as a function of
the number of uniformly reporting nodes to truly reporting nodes $(f)$
in a scale free network. The number of false positives is approximately
$f/\left(f+1\right)$. The plot was generated using Monte Carlo simulation
on a scale free network of $N=8.4K$ nodes. There are roughly $0.4N$
infected nodes, whereas only $0.1N$ nodes report their infectious
state. The degree distributation is $\Pr\left(deg(i)=x\right)=Ax^{-2.5}.$}
\end{figure}

While it is often clearer to state the proofs by means of the border
set, it is often simpler to implement the algorithm in terms of a
local ball environment. Recall the definition of the radius-$l$ ball
ball indicator

\[
H_{i}^{B}(l,s)=\begin{cases}
1 & \left|\{v\in B_{i}(l)\cap S\}\right|\geq s\\
0 & \text{otherwise}
\end{cases}.
\]

As the number of nodes in a ball may vary, and due to finite size
effects, the optimal threshold value $T$ may be different than the
corresponding theoretical value. In practice, one may optimize this
threshold value in an independent preprocessing step. In practice,
values close to the theoretical predictions were found adequate in
most cases. 
\end{document}